\documentclass[11pt]{article}
\usepackage{fullpage}
\usepackage{graphicx}
\usepackage{epstopdf}
\usepackage{amsmath}
\usepackage{amsfonts}
\usepackage{bm}
\usepackage{color}
\usepackage{subfigure}
\usepackage{amsthm}
\usepackage{tikz}
\usepackage{comment}
\usepackage[pdftex,colorlinks=true,linkcolor=blue,citecolor=red,urlcolor=blue]{hyperref}

\newcommand{\ket}[1]{\left|#1\right\rangle}
\newcommand{\braket}[2]{\left\langle#1 |  #2\right\rangle}
\newcommand{\ketbra}[2]{\left|#1 \rangle\langle #2\right|}
\newcommand{\Tr}{\textnormal{tr}}

\newtheorem{conjecture}{Conjecture}

\newtheorem{lem}{Lemma}
\newtheorem{thm}{Theorem}
\newtheorem{defi}{Definition}
\newtheorem{corollary}{Corollary}
\newtheorem{prop}{Proposition}

\begin{document}
\title{Markovian Matrix Product Density Operators : Efficient computation of global entropy}
\author{Isaac H. Kim\thanks{IBM Watson Research Center, Yorktown Heights, NY, USA 10598.}}

\maketitle

\begin{abstract}
  We introduce the Markovian matrix product density operator, which is a special subclass of the matrix product density operator. We show that the von Neumann entropy of such ansatz can be computed efficiently on a classical computer. This is possible because one can efficiently certify that the global state forms an approximate quantum Markov chain by verifying a set of inequalities. Each of these inequalities can be verified in time that scales polynomially with the bond dimension and the local Hilbert space dimension. The total number of inequalities scale linearly with the system size.

 We use this fact to study the complexity of computing the minimum free energy of local Hamiltonians at finite temperature. To this end, we introduce the free energy problem as a generalization of the local Hamiltonian problem, and study its complexity for a class of Hamiltonians that describe quantum spin chains. The corresponding free energy problem at finite temperature is in NP if the Gibbs state of such Hamiltonian forms an approximate quantum Markov chain with an error that decays exponentially with the width of the conditioning subsystem.
\end{abstract}

\section{Introduction}
Kitaev was the first to define the quantum analogue of the classical SAT problem, the local Hamiltonian problem\cite{Kitaev2002}. An instance of the $k$-local Hamiltonian problem can be viewed as a set of constraints on $n$ qubits each of which act on at most $k$ qubits. This problem is QMA-complete for $k\geq 2$ for general graphs\cite{Kempe2006}. The same conclusion holds even if the Hamiltonian consists of terms that are geometrically local on a one-dimensional(1D) lattice, provided that the local Hilbert space dimension is larger than some finite number\cite{Aharanov2009}. However, one may expect the local Hamiltonian problem to be more tractable if the ground state only has short-range correlations. Hastings was able to put this intuition on a firm footing, by showing that the problem is in NP under a promise: that there is a spectral gap which is bounded from below by a constant\cite{Hastings2007}. In fact, the same conclusion holds under a more general promise, namely that correlation decays exponentially in the ground state\cite{Brandao2014}. Furthermore, a recent breakthrough showed that there is an efficient algorithm that finds an inverse-polynomial approximation of the ground state of such systems\cite{Landau2015}, thus putting the problem in P. 

While the local Hamiltonian problem is certainly well-motivated from the perspective of the computational complexity, realistic physical systems never reach their ground state. When they are in contact with a heat bath, they eventually equilibrate to the thermal state. There are conditions under which one can show that these states only possess short-range correlations\cite{Araki1969,Kliesch2014}. These facts pose a natural question. Can we show that, in a certain sense, that a class of such Hamiltonians are ``easy'' to analyze at finite temperature? In order to formalize and make progress on this question, we propose the \emph{free energy problem.} Free energy $F$ of a quantum system described by a state $\rho$ and a Hamiltonian $H$ at temperature $T$ is defined as
\begin{equation*}
  F(\rho) = \Tr(\rho H)-TS(\rho),
\end{equation*}
where $S(\rho) := -\Tr(\rho \log \rho)$ is the von Neumann entropy of $\rho$. At zero temperature, the free energy is the ground state energy. Furthermore, the state at thermal equilibrium is \emph{defined} as the state that minimizes the free energy. Thus, it is natural to consider the following decision problem.
\begin{defi}
  The $k$-local free energy problem at a temperature $T$ is defined as follows. We are given a $k$-local Hamiltonian on $n$ qubits, with Hamiltonian $H=\sum_{i=1}^r h_i$ with $r= \text{poly}(n)$. For each $i$, $\| h_i\| \leq \text{poly}(n)$ and the entries of $h_i$ are specified by $\text{poly}(n)$ bits. We are given two constants $\alpha$ and $\beta$ such that $\beta-\alpha \geq \frac{1}{\text{poly}(n)}$. In ``YES'' instances, there exists a state $\rho$ such that $F(\rho) \leq \alpha$. In ``NO'' instances, $F(\rho)> \beta$ for every $\rho$.
\end{defi}

We believe that the free energy problem is interesting in its own right as a generalization of the local Hamiltonian problem, but there are more reasons to care. The first reason concerns simulation of quantum spin systems at finite temperature. While there are algorithms that work well in practice\cite{Verstraete2004,Zwolak2004,White2009}, it is unclear if these methods work for \emph{arbitrary} 1D systems. A recent work\cite{Barthel2017} argues that 1D quantum systems at finite temperature can be efficiently simulated on a classical computer, but this is based on an assumption that the underlying ground state can be described by a conformal field theory. This assumption would not be applicable to certain Hamiltonians, e.g., the ones described in Ref.\cite{Aharanov2009}.
  
The second reason concerns the quantum PCP conjecture\cite{Aharonov2013}. The quantum PCP conjecture states that, roughly speaking, deciding whether the ground state energy of a Hamiltonian consisting of $n$ qubits is $0$ or greater than $n\epsilon$ for some positive constant $\epsilon$ is QMA-complete. As noted in Ref.\cite{Poulin2011}, this problem is equivalent to deciding whether free energy is negative at a certain finite temperature. Thus one could disprove the conjecture by showing that there exists a classical proof that certifies the negativity of the free energy. While proving or disproving the quantum PCP conjecture remains as an outstanding open problem, the techniques that are developed to study the complexity of the free energy problem may prove useful in this context.

What makes the free energy problem so different from the ground state energy problem is that the free energy depends on the entropy. Because entropy is a nonlinear functional, it is not straightforward to conceive a method to efficiently compute it. In particular, it is worth noting that an ability to efficiently compute local expectation values does not necessarily imply that the entropy can be computed efficiently.

We could overcome this difficulty by considering a certain subclass of matrix product density operator(MPDO)\cite{Fannes1992,Verstraete2004,Zwolak2004}. MPDO is a tensor network state which is a natural generalization of the matrix product state\cite{Fannes1992,Perez-Garcia2006} to mixed states. Our key technical result is that, if a certain \emph{efficiently verifiable constraint} holds between the tensors that constitute the MPDO, the entropy of the MPDO can be decomposed into a linear combination of entropies involving MPDOs over shorter chains. This decomposition can be applied recursively, yielding an efficiently computable formula for the global entropy(up to a $\frac{1}{\text{poly}(n)}$ approximation error).

As an application, we show that, modulo a certain conjecture regarding the universal structure of thermal states for 1D local Hamiltonian, the free energy problem at finite temperature($\beta =\frac{1}{T} = O(1)$) for 1D local Hamiltonian is in NP. By 1D local Hamiltonian, we mean a Hamiltonian which is a sum of bounded terms each of which act on a finite interval on a 1D lattice. We should emphasize that the condition is formulated in a way that is independent of the ground state property. Therefore, if the conjecture is true, we should conclude that there is a fundamental difference between the 1D ground state energy problem and the 1D free energy problem at finite temperature; the former is QMA-complete, while the latter would be in NP. While we are unable to prove the conjecture at the moment, there are reasons to be optimistic. The conjecture is physically reasonable, and more importantly, it is a mild strengthening of a recently proven theorem of Kato and Brand\~ao\cite{Kato2016}.\footnote{In fact, our conjecture is implied by a conjecture put forward by these authors.} We hope a more careful analysis will yield the desired proof.

\textbf{Summary of results:} We begin by introducing a tensor network that we name as \emph{Markovian matrix product density operator}(Markovian MPDO). This is a special subclass of MPDO, and possesses an unusual property: that the von Neumann entropy of the tensor network can be computed efficiently(Theorem \ref{thm:entropy_decomposition}). This immediately leads to an efficient method to compute the free energy for any Markovian MPDO.

We then study what kind of states can be efficiently described by a Markovian MPDO. An important concept in this context is the so called \emph{uniform Markov property}.\footnote{This concept was introduced in Ref.\cite{Brandao2016}}  A system satisfies the uniform Markov property if the conditional mutual information between any two subsystems(say $A$ and $B$)  conditioned on another subsystem(say $C$) that shields $A$ from $B$ decays in the width of the subsystem $C$. We show that, for any state describing a quantum spin chain that possesses the uniform Markov property with an exponentially decaying function, there exists a polynomial-size Markovian MPDO that approximates the state up to $\frac{1}{\text{poly}(n)}$ error(Theorem \ref{thm:faithfulness}). This leads to a polynomial-size classical proof that verifies that the free energy is at most $F + \frac{1}{\text{poly}(n)}$, where $F$ is the free energy of the thermal state. In particular, if finite-temperature thermal states of quantum spin chain satisfy the uniform Markov property with an exponentially decaying function, the free energy problem for 1D quantm Hamiltonian is in NP(Theorem \ref{thm:NP}).

\textbf{A note on an important subtlety:} While discussing the content of this paper with various individuals, the author had found a recurring theme which unfortunately originates from a subtlety that is rather easy to gloss over. The main source of confusion comes from the fact that the entropy of a quantum spin chain consisting of $n$ qudits can be decomposed into a linear combination of entanglement entropies of bounded regions, provided that the conditional mutual information for certain choices of subsystems are $0$; see Ref.\cite{Poulin2011} for an example. This leads to an impression that the global entropy can be decomposed into local entropies, and thus the entropy can be computed efficiently.

This is true under a condition that the conditional mutual information vanishes, \emph{but in order to verify that the conditional mutual information vanishes, it seems necessary to compute the global entropy.} In general, without certifying the smallness of the conditional mutual information, one cannot be sure that the local decomposition is valid! Our contribution lies on obviating this verification procedure and replacing it with a condition that can be verified efficiently. In fact, we can certify the smallness of the conditional mutual information from a set of local data. This somewhat counterintuitive fact is, in our opinion, the cornerstone of our result.

\textbf{A note on a related work:} Recently, a related work has appeared\cite{Berta2017}. The way in which an approximation of the thermal state is obtained is not dissimilar. What differentiates our work is the fact that we can efficiently compute the global entropy, which leads to an efficiently computable upper bound on the free energy. 

\section{Preliminaries}
\subsection{Notation}
This paper concerns quantum spin chains. We consider a 1D finite lattice $\Lambda$ with cardinality $|\Lambda| = L$ for some nonnegative integer $L$. The lattice subsets are denoted by a set of integers from $1$ to $L$, e.g., $\{1,2,5 \}$. We shall use a short-hand notation for a range of integers. Specifically, $i,\ldots,j$ denotes integers from $i$ to $j$. The global Hilbert space associated to $\Lambda$ is $\mathcal{H}_{\Lambda} = \otimes_{i\in \Lambda} \mathcal{H}_i$, where $\dim(\mathcal{H}_i)=d < \infty$. Similarly, for any $A \subset \Lambda$, $\mathcal{H}_{A} = \otimes_{i \in A} \mathcal{H}_i$.

We will use the following set of conventions for the density matrices. A density matrix without any superscript is assumed to be in $\mathcal{B}(\mathcal{H}_{\Lambda})$. Otherwise, the support of density matrices are specified in the superscript, e.g., $\rho^{\{1,2 \}}$. In particular, if a state $\rho$ is already defined in the context, $\rho^X$ will be its reduced density matrix over a subsystem $X$. We shall consider linear maps acting on $\mathcal{B}(\mathcal{H}_{A})$, where $A \subset \Lambda$. The domain and the codomain of these maps will be specified in their subscript and superscript. For example, $\Phi_{\{1 \}}^{\{1,2 \}}$ is a map from $\mathcal{B}(\mathcal{H}_1)$ to $\mathcal{B}(\mathcal{H}_1 \otimes \mathcal{H}_2)$. There is one execption to this rule, however. An identity superoperator acting $\mathcal{B}(\mathcal{H}_A)$ is denoted as $\mathcal{I}_A$. Unless specified otherwise, these maps are assumed to be completely-positive and trace preserving(CPTP).

The convention for the norms as follows. The trace norm is written as $\| \cdot \|_1$. The operator norm is written as $\| \cdot \|$. The completely bounded norm is written as $\| \cdot\|_{\diamond}$ and defined as\cite{Kitaev2002}
\begin{equation}
  \|A \|_{\diamond} = \sup_{\rho\geq 0, \Tr(\rho)=1} \|I_k \otimes A(\rho)\|_{1},
\end{equation}
where $A$ is a linear operator acting on the space of operators and $I_k$ is the identity superoperator acting on the space of bounded operators of a $k$-dimensional Hilbert space.

\subsection{The Markov property}
The conditional mutual information between $A$ and $C$ conditioned on $B$ is defined as
\begin{equation}
  I(A:C|B)_{\rho} = S(AB)_{\rho} + S(BC)_{\rho} - S(B)_{\rho} - S(ABC)_{\rho},
\end{equation}
where $S(X)_{\rho} = -\Tr(\rho^X \log \rho^X)$ is the von Neumann entropy of the reduced density matrix of $\rho$ over a subsystem $X$. By the strong subadditivity of entropy\cite{Lieb1973}, the conditional mutual information is nonnegative.

A tripartite state $\rho^{ABC}$ is said to be an $\epsilon$-approximate quantum Markov chain if $I(A:C|B)_{\rho}\leq \epsilon$. Small conditional mutual information implies that the underlying state has the following structure\cite{Wilde2013,Berta2015,Fawzi2015}:
\begin{equation}
 \frac{1}{4\log 2} \|\rho^{ABC} - \mathcal{I}_A \otimes \Phi_B^{BC}(\rho^{AB}) \|_1 \leq I(A:C|B)_{\rho}
\end{equation}
for some CPTP map $\Phi_B^{BC}$. The converse statement is also known. That is,
\begin{equation}
  I(A:C|B) \leq 4\epsilon \log d_A + 2H_b(2\epsilon), \label{eq:FR_converse}
\end{equation}
where $\epsilon = \|\rho^{ABC} - \mathcal{I}_A \otimes \Phi_B^{BC}(\rho^{AB}) \|_1$ and $H_b(p) = -p\log p - (1-p) \log (1-p)$, assuming $\epsilon \leq \frac{1}{2}$. The converse statement says that, if a partial trace on $C$ can be recovered by a CPTP map acting on $B$, then the underlying state is an approximate quantum Markov chain. We will end up using the converse part of the statement, but with some modification.

It is clear from these results that $\epsilon$-approximate quantum Markov chains possess a rich set of structures. This warrants the notion of uniform Markov property, which was put forward in Ref.\cite{Brandao2016}. While their definition is more general, we will be content with a simplified variant.
\begin{defi}
  (1D uniform Markov property) A state $\rho$ is said to satisfy the uniform $\delta(\ell)$-Markov condition if for any $A= \{1,\ldots,a \}$, $B= \{a+1, \ldots, a+\ell \}$, $C = \{a+\ell+1,\ldots, L \}$, where $1\leq a< a+ \ell <L$,
  \begin{equation}
    I(A:C|B)_{\rho} \leq \delta(\ell).
  \end{equation}
  \label{definition:uniform_Markov}
\end{defi}
At this point, it is natural to ask whether we expect certain physical states to satisfy the $\delta(\ell)$-Markov condition, and if so, how fast the function $\delta(\ell)$ decays. Recently, Kato and Brand\~ao proved that, for the Gibbs state of 1D Hamiltonian at finite temperature, $\delta(\ell)$ decays exponentially in $\sqrt{\ell}$. The authors have also conjectured that $\delta(\ell)$ decays exponentially in $\ell$. As we shall see later, this conjecture has a nontrivial implication: that the free energy problem for 1D local Hamiltonian at finite temperature is in NP.

\subsection{Matrix product density operator}
MPDO\cite{Verstraete2004,Zwolak2004} is formally a finitely correlated state\cite{Fannes1992}. Such a state can be expressed in the following form:
\begin{equation}
(\mathcal{I}_{\{1,\ldots,n-2 \}} \otimes \Phi_{\mathfrak{B}_{n-1}}^{\mathfrak{B}_n \cup \{n-1 \}}) \circ \cdots \circ(\mathcal{I}_{\{1 \}} \otimes \Phi_{\mathfrak{B}_2}^{\mathfrak{B}_3 \cup \{2 \}})\circ \Phi_{\mathfrak{B}_1}^{\mathfrak{B}_2 \cup \{1 \}} (\sigma^{\mathfrak{B}_1}),\label{eq:MPDO}
\end{equation}
where $\sigma^{\mathfrak{B}_1}$ is a density matrix that acts on $\mathcal{H}_{\mathfrak{B}_1}$. The subsystems $\mathfrak{B}_i$ are what is known as the ``virtual space'' in the langauge of tensor networks. Note that Eq.\ref{eq:MPDO} acts on $\mathcal{H}_{\mathfrak{B}_n} \otimes \mathcal{H}_{\{1,\ldots,n-1 \}}$ and  not on $\mathcal{H}_{\{ 1,\ldots, n-1\}}$. Often in practice a partial trace on $\mathfrak{B}_n$ is carried out because $\mathcal{H}_{\mathfrak{B}_n}$ is not part of the physical Hilbert space. However, we shall leave it as is to define a special subclass of MPDO: Markovian MPDO.

It is well-known that MPDO can be contracted in time that scales polynomially with the dimension of $\max_i \dim (\mathcal{H}_{\mathfrak{B}_i})$, $d,$ and $n$\cite{Verstraete2004,Zwolak2004}. The \emph{bond dimension}, which in our context is $\dim (\mathcal{H}_{\mathfrak{B}_i})$, will play a particularly important role later.

\section{Markovian MPDO}
Markovian MPDO is a MPDO that satisfies a certain set of constraints.
\begin{defi}
  A state over $n$ qudits, denoted as $\rho$, is an $\epsilon$-approximate Markovian MPDO if
  \begin{equation}
    \rho = (\mathcal{I}_{\{1,\ldots,n-2\}}\otimes  \Phi_{\{n-1\}}^{\{n-1,n\}})\circ \cdots \circ (\mathcal{I}_{\{1 \}} \otimes \Phi_{\{2\}}^{\{2,3\}}) \circ \Phi_{\{1\}}^{\{1,2\}}(\tilde{\rho}^{(1)}),
  \end{equation}
    such that
  \begin{equation}
    \|\Phi_{\{i-1\}}^{\{i-1,i\}}(\tilde{\rho}^{\{i-1\}}) - \Tr_{i+1}[(\mathcal{I}_{\{ i-1\}} \otimes \Phi_{\{i\}}^{\{i,i+1\}})\circ \Phi_{\{i-1\}}^{\{i-1,i\}}(\tilde{\rho}^{\{i-1\}}) ] \|_1 \leq \epsilon
  \end{equation}
  for $i=2, \cdots, n-1$, where
  \begin{equation}
    \tilde{\rho}^{\{i\}} := \Tr_{1,\cdots, i-1}[( \mathcal{I}_{\{1,\ldots,i-2 \}} \otimes \Phi_{\{i-1\}}^{\{i-1,i\}})\circ\cdots\circ\Phi_{\{1\}}^{\{1,2\}}(\tilde{\rho}^{\{1\}})]
  \end{equation}
\end{defi}

A straightforward, yet nevertheless an important property of Markovian MPDO is that one can efficiently verify that a MPDO is a Markovian MPDO.
\begin{prop}
  Suppose $\rho$ is an $\epsilon$-approximate Markovian MPDO over $n$ qudits with local qudit dimension $d$, where each of the entries for $\Phi_{\{i-1\}}^{\{i-1,i \}}$, $1<i \leq n$ and $\tilde{\rho}^{\{1 \}}$ are specified with $\text{poly}(n)$ bits. Then there is an algorithm that runs in $\text{poly}(n,d)$ time that verifies $\rho$ is an $\epsilon$-approximate Markovian MPDO.
  \label{proposition:MMPDO_verification}
\end{prop}
\begin{proof}
  Note that
  \begin{equation}
    \begin{aligned}
      \tilde{\rho}^{\{i\}} &= \Tr_{1,\cdots, i-1}[(\mathcal{I}_{\{1,\ldots, i-2 \}}\otimes \Phi_{\{i-1\}}^{\{i-1,i\}})\circ\cdots\circ\Phi_{\{1\}}^{\{1,2\}}(\tilde{\rho}^{\{1\}})] \\
      &= (\Tr_{i-1} \circ  \Phi_{\{i-1 \}}^{\{i-1,i\}}) \circ \cdots (\Tr_{1} \circ \Phi_{\{1 \}}^{\{ 1,2\}}) (\tilde{\rho}^{ \{1 \}}).
    \end{aligned}
  \end{equation}
  The entries of the linear maps $\Tr_{j} \circ \Phi_{\{j \}}^{\{ j, j+1\}}$ can be computed in $\text{poly}(n,d)$ time because both $\Tr_{j}$ and $\Phi_{\{j \}}^{\{j,j+1 \}}$ act on $\text{poly}(d)$-dimensional space and their entries are specified with $\text{poly}(n)$ precision. Therefore, $\tilde{\rho}^{\{ i\}}$ can be computed in $\text{poly}(n,d)$ time. For the same reason, $\|\Phi_{\{i-1\}}^{\{i-1,i\}}(\tilde{\rho}^{\{i-1\}}) - \Tr_{i+1}[(\mathcal{I}_{\{i-1 \}} \otimes\Phi_{\{i\}}^{\{i,i+1\}})\circ \Phi_{\{i-1\}}^{\{i-1,i\}}(\tilde{\rho}^{\{i-1\}}) ] \|_1$ can be computed in $\text{poly}(n,d)$ time.
\end{proof}

\subsection{Local certificate for approximate Markov chain}
Let us recall the meaning of Eq.\ref{eq:FR_converse}. It says that, if there exists a quantum channel that acts on $B$ such that it can recover a partial trace of $C$ over a tripartite state $\rho^{ABC}$, then $\rho^{ABC}$ forms an approximate quantum Markov chain.

Here we derive an inequality which has a markedly different meaning, yet technically looks very similar to Eq.\ref{eq:FR_converse}. We would like to compare this meaning to the meaning of Eq.\ref{eq:FR_converse} in order to prevent a potential confusion. Lemma \ref{lemma:QMC_local_characterization} says that, given a state $\sigma^{AB}$, if there exists a channel $\Phi_B^{BC}$ such that its restricted action on $B$ does not alter $\sigma^{AB}$ too much, then $\rho^{ABC} = \mathcal{I}_A \otimes \Phi_B^{BC}(\sigma^{AB})$ forms an approximate quantum Markov chain.

The key difference between these two statement lies on the condition that implies the Markov property. In the former, the condition involves the closeness of two tripartite states, i.e., $\rho^{ABC} \approx \mathcal{I}_A \otimes \Phi_B^{BC} (\rho^{AB})$. In the latter, the condition involves closeness of two \emph{bipartite states,} i.e., $\rho^{AB} \approx \sigma^{AB}$. This difference is important because, in order to prove our main result, we will need to be able to efficiently certify that a given state satisfies the Markov property. If we were to use the former statement to certify, we will need to compute the trace distance between two density matrices supported on $ABC$. If the dimension of $C$ is large, it is generally not obvious how to do that efficiently. On the other hand, computing the trace distance between two density matrices supported on $AB$ can be done efficiently as long as the dimension of $AB$ is sufficiently small. 

\begin{lem}
  Suppose $\rho^{ABC} = \mathcal{I}_A \otimes \Phi_B^{BC}(\sigma^{AB})$, where $\Phi_B^{BC} : \mathcal{B}(\mathcal{H}_B) \to \mathcal{B}(\mathcal{H}_{B} \otimes \mathcal{H}_C)$ is a CPTP map such that $\|\sigma^{AB} - \rho^{AB} \|_1 \leq \epsilon$ and $\epsilon\leq \frac{1}{2}$. Then
  \begin{equation}
    I(A:C|B)_{\rho} \leq 4\epsilon \log d_A + 2H_b(2\epsilon),
  \end{equation}
  where $H_b(p) = -p\log p - (1-p) \log(1-p)$ is the binary entropy function.
  \label{lemma:QMC_local_characterization}
\end{lem}
\begin{proof}
  \begin{equation}
    \begin{aligned}
      I(A:C|B)_{\rho} = S(A|B)_{\rho} - S(A|BC)_{\rho},
    \end{aligned}
  \end{equation}
  where $S(A|B)_{\rho} = S(AB)_{\rho} - S(B)_{\rho}$ is the quantum conditional entropy. Note that
  \begin{equation}
    \begin{aligned}
      S(A|B)_{\rho} &= -D(\rho^{AB} \| \frac{I_A}{d_A} \otimes \rho^B) +\log d_A\\
      &\leq -D(\mathcal{I}_A \otimes \Phi_B^{BC}(\rho^{AB}) \| \frac{I_A}{d_A} \otimes \Phi_B^{BC}(\rho^B))  + \log d_A\\
      &= S(A|BC)_{\rho'},
    \end{aligned}
  \end{equation}
  where $\rho' = \mathcal{I}_A \otimes \Phi_B^{BC}(\rho^{AB})$ and $D(\rho\|\sigma):= \Tr(\rho (\log \rho - \log \sigma))$ is the quantum relative entropy. Here, the inequality follows from the monotonicity of quantum relative entropy under CPTP maps.

Also, note that
\begin{equation}
  \begin{aligned}
    \|\rho'^{ABC} - \rho^{ABC} \|_1 &= \|\mathcal{I}_A \otimes\Phi_B^{BC}( \rho^{AB}) - \rho^{ABC}\|_1 \\
    &= \|\mathcal{I}_A\otimes \Phi_B^{BC}(\rho^{AB} - \sigma^{AB})\|_1 \\
    &\leq \|\rho^{AB} - \sigma^{AB} \|_1,
  \end{aligned}
\end{equation}
where we used the fact that (i) $\|\Phi(\ldots)\|_1\leq \|\ldots \|_1$ for any CPTP map $\Phi$ and (ii) $\rho^{ABC} = \mathcal{I}_A \otimes \Phi_B^{BC}(\sigma^{AB})$.

Using Fannes-Alicki inequality\cite{Alicki2004}, 
  \begin{equation}
    \begin{aligned}
      I(A:C|B)_{\rho} &\leq S(A|BC)_{\rho'} - S(A|BC)_{\rho} \\
      &\leq 4\epsilon \log d_A - 4\epsilon \log(2\epsilon) - 2(1-2\epsilon) \log(1-2\epsilon).
    \end{aligned}
  \end{equation}
\end{proof}

\subsection{Local entropy decomposition}
Markovian MPDOs are unusual in that one can efficiently and accurately compute the global entropy, provided that $\epsilon= 1/\text{poly}(n)$. Below, we derive an efficiently computable formula for the entropy, with a rigorous stability bound.

\begin{lem}
  For any $\epsilon$-approximate($\epsilon \leq \frac{1}{2}$) Markovian MPDO $\rho$ over $n$ qudits with qudit dimension $d$,
  \begin{equation}
    |S(\rho) - S({\{1,2\}})_{\rho} - S({\{2,\ldots,n\}})_{\rho} + S({\{2\}})_{\rho}| \leq 4\epsilon \log d + 2H_b(2\epsilon). \label{eq:recursion}
  \end{equation}
  \label{lemma:recursion}
\end{lem}
\begin{proof}
  Let us apply Lemma \ref{lemma:QMC_local_characterization} by setting $A=\{1 \}$, $B= \{2 \}$, and $C= \{3, \ldots, n \}$. Note that
  \begin{equation}
    \rho^{ABC} = \mathcal{I}_A \otimes \Phi_B^{BC}(\sigma^{AB}),
  \end{equation}
  where $\Phi_B^{BC} = (\mathcal{I}_{\{1,\ldots,n-2 \}}\otimes\Phi_{\{n-1\}}^{\{n-1,n\}}) \circ \ldots \circ \Phi_{\{2\}}^{\{2,3\}}$ and $\sigma^{AB} = \Phi_{\{1 \}}^{\{1,2 \}}(\tilde{\rho}^{\{1 \}})$. Because $\Tr_{i,i+1}[\Phi_{\{i\}}^{\{i,i+1\}}(O)] = \Tr_{i,i+1}[O]$ for any operator $O$, 
  \begin{equation}
    \rho^{AB} =  \Tr_3[(\mathcal{I}_{\{1 \}}\circ \Phi_{\{2\}}^{\{2,3\}}) \circ \Phi_{\{1\}}^{\{1,2\}}(\tilde{\rho}^{\{1\}})].
  \end{equation}
  Because $\rho$ is an $\epsilon$-approximate Markovian MPDO,
  \begin{equation}
    \|\rho^{AB} - \Phi_{\{1\}}^{\{1,2\}}(\tilde{\rho}^{\{1\}}) \|\leq \epsilon.
  \end{equation}
  Therefore, by Lemma \ref{lemma:QMC_local_characterization},
  \begin{equation}
    I(A:C|B)_{\rho} \leq 4\epsilon \log d + 2H_b(2\epsilon).
  \end{equation}
  By the strong subadditivity of entropy\cite{Lieb1973}, $I(A:C|B)_{\rho} \geq 0$. Combining these two bounds, Eq.\ref{eq:recursion} is derived.
\end{proof}

By definition, $\rho^{\{2,\ldots,n\}}$ is an $\epsilon-$approximate Markovian MPDO over $n-1$ qudits if  $\rho$ is an $\epsilon-$approximate Markovian MPDO over $n$ qudits. Therefore, one can apply Lemma \ref{lemma:recursion} recursively to prove the following decomposition of the global entropy.
\begin{thm}
  For any $\epsilon$-approximate($\epsilon \leq \frac{1}{2}$) Markovian MPDO $\rho$ over $n\geq 3$ qudits,
  \begin{equation}
    |S(\rho) - (\sum_{i=2}^{n} S({\{i-1,i\}})_{\rho}) - (\sum_{i=2}^{n-1} S({\{i\}})_{\rho})| \leq (n-2)(4\epsilon \log d + 2H_b(2\epsilon)). \label{eq:entropy_decomposition}
  \end{equation}
  \label{thm:entropy_decomposition}
\end{thm}
\begin{proof}
  By Lemma \ref{lemma:recursion}, for $ 1 \leq i\leq n-2$,
  \begin{equation}
    |S({\{i,\ldots, n\}})_{\rho} - S({\{i,i+1\}})_{\rho} - S({\{i+1,\ldots,n\}})_{\rho} + S({\{i+1\}})_{\rho}| \leq 4\epsilon \log d + 2H_b(2\epsilon).
  \end{equation}
  Applying this bound for all $1\leq i \leq n-2$ and using the triangle inequality, Eq.\ref{eq:entropy_decomposition} is derived.
\end{proof}

\subsection{Variational upper bound on free energy}
As a corollary of Theorem \ref{thm:entropy_decomposition}, we can obtain a variational upper bound to the free energy.
\begin{corollary}
  Consider a Hamiltonian acting on $n$ qudits. Its free energy at temperature $T$ is bounded by
\begin{equation}
  \min_{\rho \geq 0, \Tr(\rho)= 1}F(\rho) \leq \Tr[\rho H] - T((\sum_{i=2}^n S({\{i-1,i\}})_{\rho}) - (\sum_{i=2}^{n-1}S({\{i\}})_{\rho})) + (n-2)(4\epsilon \log d + 2H_b(2\epsilon)), \label{eq:free_energy_bound}
\end{equation}
for any $\epsilon-$approximate($\epsilon \leq \frac{1}{2}$) Markovian MPDO $\rho$, where $d$ is the qudit dimension.
\end{corollary}
Given a Markovian MPDO $\rho$ this upper bound can be computed efficiently as long as (i) $H$ consists of $\text{poly}(n)$ terms which are $O(\log n)$-local and (ii) $d = \text{poly}(n)$. Because any MPDO is an $\epsilon-$approximate MPDO for a suitably large $\epsilon$, the bound can be of course computed for any MPDO. However, the last term in Eq.\ref{eq:free_energy_bound} will generally make the bound less tight than desired.

\section{Complexity of 1D free energy problem \label{section:complexity}}
Now we define the 1D free energy problem, and show that the problem is in NP if the Gibbs state $\rho = e^{-\beta H} / \Tr(e^{-\beta H})$ of a 1D local Hamiltonian $H$ at inverse temperature $\beta = \frac{1}{T} = O(1)$ obeys the uniform Markov property(Definition \ref{definition:uniform_Markov}) with exponentially decaying function. We call $H$ as a 1D local Hamiltonian over $n$ qubits if $H= \sum_{i=1}^{n} h_i$ where $h_i$ are hermitian operators acting nontrivially on a ball of radius $O(1)$ centered at $i$ and $\|h_i\| \leq 1$; the underlying lattice is a 1D lattice.
\begin{defi}
(1D free energy problem)  Given a 1D local Hamiltonian $H$ over $n$ qubits, inverse temperature $T =O(1)$ and $\alpha,\beta$ such that $\beta- \alpha \geq\frac{1}{\text{poly}(n)}$, we have a promise that $F=\min_{\rho\geq 0, \Tr\rho = 1}[E-TS]\leq \alpha$ or $F>\beta$. The problem is to decide if $F\leq \alpha$. When $F\leq \alpha$, we say we have a YES instance. If $F>\beta$, we we have NO instance.
\end{defi}

Let us put forward the following conjecture; see also Ref.\cite{Kato2016} for a more general formulation of the conjecture.
\begin{conjecture}
  There are universal constants $\ell_0, c,$ and $a$ such that for any 1D local Hamiltonian over $n$ qubits, for any inverse temperature $\beta >0$, and for any three contiguous subsystems $A,B,$ and $C$, 
  \begin{equation}
    I(A:C|B)_{\rho}\leq e^{-a\exp(-c\beta) d(A,C)},
  \end{equation}
  for any $d(A,C) \geq \ell_0$, where $d(A,C)$ is the distance between $A$ and $C$ and $\rho = e^{-\beta H} / \Tr[e^{-\beta H}]$
  \label{conjecture:Markov}
\end{conjecture}

\subsection{Efficient Markovian MPDO description of the Gibbs state}
The first implication of Conjecture \ref{conjecture:Markov} is that the Gibbs state of 1D local Hamiltonian at finite temperature has an efficient Markovian MPDO description. The main argument is based on the recovery channel of Fawzi and Renner\cite{Fawzi2015}.
\begin{lem}
  Under Conjecture \ref{conjecture:Markov}, for any Gibbs state $\rho$ of 1D local Hamiltonian at finite temperature, for any three contiguous subsystems $A, B,$ and $C$, there exists a CPTP map $\Phi_B^{BC} : \mathcal{B}(\mathcal{H}_B) \to \mathcal{B}(\mathcal{H}_B \otimes \mathcal{H}_C)$ such that
  \begin{equation}
    \|\rho^{ABC} - \mathcal{I}_A \otimes \Phi_B^{BC}(\rho^{AB}) \|_1 \leq 2 e^{-\frac{a}{2}\exp(-c\beta) d(A,C)}.
  \end{equation}
  and $\Phi_B^{BC}(\rho^{B}) = \rho^{BC}$.
  \label{lemma:recoverability}
\end{lem}
\begin{proof}
  For any tripartite state $\rho^{ABC}$, there exists a CPTP map $\Phi_B^{BC}: \mathcal{B}(\mathcal{H}_B) \to \mathcal{B}(\mathcal{H}_B \otimes \mathcal{H}_C)$\cite{Fawzi2015} such that
  \begin{equation}
    \frac{1}{4} \|\rho^{ABC} - \mathcal{I}_A \otimes \Phi_B^{BC}(\rho^{AB})\|_1^2\leq I(A:C|B)_{\rho}
  \end{equation}
  and $\Phi_B^{BC}(\rho^B) = \rho^{BC}$.  The inequality follows from the upper bound on $I(A:C|B)_{\rho}$.
\end{proof}
 
Below, we show that the Gibbs state of 1D local Hamiltonian indeed has an efficient Markovian MPDO description. This is done by constructing a MPDO which approximates the thermal state, and then showing that it is in fact a Markovian MPDO.
\begin{thm}
  Under Conjecture \ref{conjecture:Markov}, for every 1D local Hamiltonian consisting of $n$ qubits, and at every temperature $\beta=\frac{1}{T} = O(1)$, for every $\eta>0 $, there exists a $O(\frac{1}{n^{\eta}})$-approximate Markovian MPDO $\rho_{\text{MPDO}}$ with $O(n^{2\eta e^{c\beta}/a })$ bond dimension and qudit dimension such that
  \begin{equation}
    \|e^{-\beta H}/ \Tr(e^{-\beta H})  - \rho_{\text{MPDO}}\|_1 =  O(\frac{1}{n^{\eta -1 }})
  \end{equation}
  \label{thm:faithfulness}
\end{thm}
\begin{proof}
  Choose an integer $L_0= \left \lceil{b\log n}\right \rceil $ such that $b > 2e^{c\beta}/a$. Let $\eta=\frac{ab}{2e^{c\beta}}>2 $. Relabel the qubits $(j-1)L_0 +1, (j-1)L_0 +2, \ldots, (j-1)L_0 + jL_0$ as the $j$th qudit for $jL_0 < n$, and relabel the remaining qubits as the $\left \lceil{\frac{n}{L_0}} \right \rceil$th qudit.  Then the Hamiltonian can be expressed as $H=\sum_{i=1}^{m} h_i'$, where $m=O(n/\log n)$, $h_i'$ is a hermitian operator acting on a ball of radius $O(1)$ with qudit dimension $d=\text{poly}(n)$, and $\| h_i' \| = O(\log n)$.

  Let $\rho=e^{-\beta H}/\Tr[e^{-\beta H}]$. For any $\beta =O(1)$ and any $i$, Lemma \ref{lemma:recoverability} implies that there exists a CPTP map $\Lambda_{(i)}^{(i,\ldots, m)}$ such that
  \begin{equation}
    \|\rho -  \mathcal{I}_{\{1,\ldots, i-1\}} \otimes \Lambda_{\{i\}}^{\{i,\ldots, m\}}(\rho^{\{1,\ldots,i\}}) \|_1 \leq 2n^{-\eta} \label{eq:MarkovBound}
  \end{equation}
and
\begin{equation}
  \Lambda_{\{i \}}^{\{i,\dots,m \}}(\rho^{\{ i \}}) = \rho^{\{i,\ldots, m\}}.\label{eq:MarkovEquality}
  \end{equation}
By Eq.\ref{eq:MarkovBound},
  \begin{equation}
    \begin{aligned}
      \|\rho^{\{1,\ldots,i\}} - \Tr_{i+1,\ldots,m}[\mathcal{I}_{\{1,\ldots,i-1 \}}\otimes \Lambda_{\{i-1\}}^{\{i-1,...,m\}}(\rho^{\{1,\ldots,i-1\}}) ]  \|_1
      &\leq     \|\rho -  \mathcal{I}_{\{1,\ldots,i-1 \}} \otimes \Lambda_{\{i-1\}}^{\{i-1,\ldots, m\}}(\rho^{\{1,\ldots,i-1\}}) \|_1  \\
      &\leq 2n^{-\eta}.
    \end{aligned}
  \end{equation}
  In other words, for all $i$, there exists a CPTP map $\Phi_{\{i-1\}}^{\{i-1,i\}} = \Tr_{i+1,\ldots,m}\circ \Lambda_{\{i-1\}}^{\{i-1,\ldots,m\}}$ such that $\|\rho^{\{1,...,i\}}- \mathcal{I}_{\{1,\ldots, i-2 \}}\otimes \Phi_{\{i-1\}}^{\{i-1,i\}}(\rho^{\{1,\ldots, i-1\}})\|\leq 2n^{-\eta}$. Therefore, using the fact that $m\leq n$,
  \begin{equation}
    \|\rho - \rho_{\text{MPDO}} \|_1 \leq n^{1-\eta},
  \end{equation}
  where
  \begin{equation}
    \rho_{\text{MPDO}} = (\mathcal{I}_{\{1,\ldots,m-2\}}\otimes \Phi_{\{m-1\}}^{\{m-1,m\}})\circ \ldots \circ\Phi_{\{1\}}^{\{1,2\}}(\tilde{\rho}_{\text{MPDO}}^{\{1\}}).
  \end{equation}
  Here we set $\tilde{\rho}_{\text{MPDO}}^{\{1 \}} = \rho^{\{1 \}}$, $\Phi_{\{1\}}^{\{1,2\}}(\tilde{\rho}_{\text{MPDO}}^{\{1\}}) = \rho^{\{1,2\}}$, and $\Phi_{\{i\}}^{\{i,i+1\}} = \Tr_{i+1,\ldots,m}\circ \Lambda_{\{i-1\}}^{\{i-1,\ldots,m\}}$ for $i\geq 2.$

  Now we have all the necessary ingredients for the proof. Let us first note that $\rho_{\text{MPDO}}$ is a MPDO with a bond dimension of $\text{poly}(n)$, provided that $b=O(1)$. Therefore, it suffices to show that
  \begin{equation}
    \|\Phi_{\{i-1\}}^{\{i-1,i\}}(\tilde{\rho}_{\text{MPDO}}^{\{i-1\}}) - \Tr_{i+1}[(\mathcal{I}_{\{ i-1\}} \otimes \Phi_{\{i\}}^{\{i,i+1\}})\circ \Phi_{\{i-1\}}^{\{i-1,i\}}(\tilde{\rho}_{\text{MPDO}}^{\{i-1\}}) ] \|_1 = O(\frac{1}{n^{\eta}})
  \end{equation}
  for $i=2, \cdots, n-1$, where
  \begin{equation}
    \tilde{\rho}_{\text{MPDO}}^{\{i\}} := \Tr_{1,\cdots, i-1}[( \mathcal{I}_{\{1,\ldots,i-1 \}} \otimes \Phi_{\{i-1\}}^{\{i-1,i\}})\circ\cdots\circ\Phi_{\{1\}}^{\{1,2\}}(\tilde{\rho}_{\text{MPDO}}^{\{1\}})]
  \end{equation}
  while keeping $b=O(1)$.
  
  Let us first note that $\tilde{\rho}_{\text{MPDO}}^{\{ i\}} = \rho^{\{i \}}$. This is proved by induction. The claim is trivially true for $i=1,2$. For $i>3$, 
  \begin{equation}
    \begin{aligned}
      \tilde{\rho}_{\text{MPDO}}^{\{i \}} &= \Tr_{i-1}[ \Phi_{\{i-1 \}}^{\{i-1,i \}}(\tilde{\rho}_{\text{MPDO}}^{\{i-1 \}})]\\
      &= \Tr_{i-1}[ \Phi_{\{i-1 \}}^{\{i-1,i \}}(\rho^{\{i-1 \}})] \\
      &= \Tr_{i-1, i+1,\ldots, m}[\Lambda_{\{i-1 \}}^{\{i-1,\ldots, m \}}(\rho^{\{i-1 \}})] \\
      &= \rho^{\{i \}},
    \end{aligned}
  \end{equation}
  where we used Eq.\ref{eq:MarkovEquality}

  Furthermore,
  \begin{equation}
    \begin{aligned}
      &\|\rho^{\{i-1,i \}} -  \Tr_{i+1}[(\mathcal{I}_{\{ i-1\}} \otimes \Phi_{\{i\}}^{\{i,i+1\}})\circ \Phi_{\{i-1\}}^{\{i-1,i\}}(\tilde{\rho}_{\text{MPDO}}^{\{i-1\}}) ]\|_1 \\ &= \|\rho^{\{i-1,i\}} - \Tr_{i+1}[(\mathcal{I}_{\{ i-1\}} \otimes \Phi_{\{i\}}^{\{i,i+1\}})(\rho^{\{i-1,i \}}) \|_1 \\
      &= \|\rho^{\{i-1,i \}} - \Tr_{i+1,\ldots, m}[\mathcal{I}_{\{i-1 \}} \otimes \Lambda_{\{i \}}^{\{i,\ldots,m \}}(\rho^{\{i-1,i \}})] \|_1 \\
      &\leq \|\rho - \mathcal{I}_{\{1,\ldots,i-1 \}}\otimes \Lambda_{\{i \}}^{\{i,\ldots,m \}}(\rho^{\{1,\ldots,i \}}) \|_1\\
      &\leq n^{-\eta},
    \end{aligned}
  \end{equation}
  where we used Eq.\ref{eq:MarkovBound} in the last line.

  Therefore, there exists an $O(\frac{1}{n^{\eta}})$-approximate Markovian MPDO $\rho_{\text{MPDO}}$ with $O(n^{b})$-bond dimension and qudit dimension, such that $\|\rho - \rho_{\text{MPDO}} \|_1 =O(n^{1-\eta})$ where $\eta = ab/(2e^{c\beta})$.
\end{proof}

We would like to point out that Kato and Brand\~ao proved a theorem that is slightly weaker than Conjecture \ref{conjecture:Markov}. The main difference is that we require the Markov chain property to hold with an exponentially decaying error, whereas these authors proved the property with a subexponentially decaying error. It will be interesting to be able to prove our conjecture by sharpening the tools developed by these authors.

\subsection{Finite precision approximation}
It is important to note that Theorem \ref{thm:faithfulness} cannot be directly used to show that 1D free energy problem is in NP under Conjecture \ref{conjecture:Markov}. The reason is that the CPTP maps that define $\rho_{\text{MPDO}}$ may require an exponentially long precision. In order to ensure that polynomial number of bits of information is sufficient, we need to invoke a few elementary facts.

\begin{lem}
  For any $d\times d$ density matrix $\rho$, there exists a density matrix $\tilde{\rho}$, specified with $O(d^2n)$ bits, such that $\|\rho - \tilde{\rho} \|_1 \leq O(d^22^{-n})$. Furthermore, the fact that $\tilde{\rho}$ is a normalized density matrix canbe verified in $\text{poly}(n,d)$ time. \label{lem:finite_precision}
\end{lem}
\begin{proof}
  Without loss of generality, let $\rho = \sum_j p_j \ketbra{\psi_j}{\psi_j}$, where $\ket{\psi_j}$ are the eigenstates of $\rho$ with eigenvalues $p_j$ such that $p_{j} \leq p_{j+1}$. Also, let $\ket{\psi_j} = \sum_{k=1}^{d} a_{jk} \ket{k}$ for complex numbers $a_{jk}$ and some basis $\{ \ket{j}\}$.  Consider the following subnormalized states
  \begin{equation}
    \ket{\tilde{\psi}_j} = \sum_{k=1}^{d} \tilde{a}_{jk}\ket{k},
  \end{equation}
  where $\tilde{a}_{jk} = f_n(\Re(a_{jk})) + i f_n(\Im(a_{jk}))$ and $f_n(x)$ is a function that keeps the first $n$ bits of a real number of $x$. The norm of $\ket{\tilde{\psi}_j}$ obeys the following bound:
  \begin{equation}
    1- d 4^{-n}\leq \braket{\tilde{\psi}_j}{\tilde{\psi}_j} \leq 1.
  \end{equation}
  Now, let $\tilde{p}_j = f_n(p_j)$ for $j<d$ and $\tilde{p}_d$ is the following rational number:
  \begin{equation}
    \tilde{p}_d = 1-\sum_{k=1}^{d-1} \tilde{p}_j /\braket{\tilde{\psi}_j}{\tilde{\psi}_j}.
  \end{equation}
  Setting $\tilde{\rho} = \sum_{j} \tilde{p}_j \ketbra{\tilde{\psi}_j}{\tilde{\psi}_j}$,
  \begin{equation}
    \begin{aligned}
    \|\rho - \tilde{\rho} \|_1 &\leq \sum_j \|p_j \ketbra{\psi_j}{\psi_j} - \tilde{p}_j \ketbra{\tilde{\psi}_j}{\tilde{\psi}_j} \|_1 \\
    &\leq \sum_j p_j \|\ketbra{\psi_j}{\psi_j} - \ketbra{\tilde{\psi}_j}{\tilde{\psi}_j} \|_1 + \sum_j |p_j - \tilde{p}_j| \\
    &\leq \sum_{j,k,l}j p_j |a_{jk}a^*_{jl} - \tilde{a}_{jk}\tilde{a}^*_{jl}| + \sum_j |p_j - \tilde{p}_j| \\
    &\leq d^2(2^{1-n} + 4^{-n} ) + 2^{1-n} \\
    &\leq 5d^2 2^{-n}.
  \end{aligned} 
\end{equation}
  The fact that $\tilde{\rho}$ is positive semi-definite is manifest from the construction, and the fact that it is normalized can be verified efficiently because all the entries are rational numbers specified up to $n$ bits.
\end{proof}

By using the Jamio\l{}kowski isomorphism\cite{Jamiolkowski1972, Choi1975}, a CPTP map can be specified by a state(on an extended Hilbert space) and vice versa. Specifically, given a CPTP map $\Phi: \mathcal{B}(\mathcal{H}_A) \to \mathcal{B}(\mathcal{H}_B)$, its Jamio\l{}kowski state $\rho_{\Phi}$ is defined as
\begin{equation}
  \rho_{\Phi} = \mathcal{I}_{A'} \otimes \Phi(\ketbra{\Psi_{A'A}}{\Psi_{A'A}})
\end{equation}
where $\ket{\Psi_{A'A}}$ is a maximally entangled state between $A'$ and $A$ with $\dim A' = \dim A$. The following bound is known\cite{Nechita2016}(Proposition 1):
\begin{equation}
  \|\Phi - \Phi' \|_{\diamond} \leq  \|\Tr_B |\rho_{\Phi} - \rho_{\Phi'}| \|,
\end{equation}
where $\Phi,\Phi' : \mathcal{B}(\mathcal{H}_A) \to \mathcal{B}(\mathcal{H}_B)$ are CPTP maps, $\| \cdot \|_{\diamond}$ is the diamond norm\cite{Kitaev2002}, and $\|\cdot \|$ is the operator norm. Here $| O | = V|D|V^{\dagger}$ for any hermitian operator which can be diagonalized as $O=VDV^{\dagger}$, where $|D|$ is an entry-wise absolute value of $D$. This bound implies that
\begin{equation}
  \|\Phi - \Phi' \|_{\diamond} \leq d_A \|\rho_{\Phi} - \rho_{\Phi'} \|_1. \label{eq:channel_state_bound}
\end{equation}

\begin{thm}
  Under Conjecture \ref{conjecture:Markov}, for every 1D local Hamiltonian consisting of $n$ qubits, and at every temperature $\beta = \frac{1}{T} = O(1)$, for every $\eta>0$, there exists a $O(\frac{1}{n^{\eta}})$-approximate Markovian MPDO $\tilde{\rho}_{\text{MPDO}}$ with $O(n^{2\eta e^{c\beta}/a})$ bond dimension and qudit dimension with $\text{poly}(n)$ precision such that
  \begin{equation}
    \|e^{-\beta H}/ \Tr(e^{-\beta H}) - \tilde{\rho}_{\text{MPDO}} \|_1 = O(\frac{1}{n^{\eta - 1}}).
  \end{equation}
  \label{thm:finite_precision}
\end{thm}
\begin{proof}
  By Theorem \ref{thm:faithfulness}, there exists a polynomial-size bond dimension Markovian MPDO $\rho_{\text{MPDO}}$ such that
  \begin{equation}
    \|e^{-\beta H} / \Tr(e^{-\beta H}) - \rho_{\text{MPDO}} \|_1 = O(\frac{1}{n^{\eta -1}}).
  \end{equation}
  Let us define $\tilde{\rho}_{\text{MPDO}}$ as
  \begin{equation}
    \tilde{\rho}_{\text{MPDO}} = (\mathcal{I}_{\{1,\ldots, n-2 \}} \otimes \tilde{\Phi}_{\{n-1 \}}^{\{n-1,n \}})\circ \cdots \circ \tilde{\Phi}_{\{ 1 \}}^{\{ 1,2 \}}(\tilde{\rho}^{(1)}),
  \end{equation}
  where the CPTP maps and the operators are defined in terms of the CPTP maps and the operators that define $\rho_{\text{MPDO}}$. Specifically, let
  \begin{equation}
    \rho_{\text{MPDO}} = (\mathcal{I}_{\{1,\ldots, n-2 \}} \otimes \Phi_{\{n-1 \}}^{\{n-1,n \}})\circ \cdots \circ \Phi_{\{ 1 \}}^{\{ 1,2 \}}(\rho^{(1)}).
  \end{equation}
  Because of Lemma \ref{lem:finite_precision}, there exists a density matrix $\tilde{\rho}^{(1)}$ which is specified with $\text{poly}(n)$ bits such that $\|\rho^{(1)} - \tilde{\rho}^{(1)} \|_1 \leq O(e^{-\text{poly}(n)})$. For each $\Phi_{\{i-1 \}}^{\{i-1,i \}}$, let $\rho_{\Phi_{\{i-1 \}}^{\{i-1,i \}}}$ be its Jamio\l{}kowski state. Again by Lemma \ref{lem:finite_precision} there exists a state $\tilde{\rho}_{\Phi_{\{i-1 \}}^{\{i-1,i \}}}$ specified with $\text{poly}(n)$ precision such that $\|\rho_{\Phi_{\{i-1 \}}^{\{i-1,i \}}} - \tilde{\rho}_{\Phi_{\{i-1 \}}^{\{i-1,i \}}}\|_1 = O(e^{-\text{poly}(n)}).$ By Eq.\ref{eq:channel_state_bound}, this bound implies that $\|\Phi_{\{i-1 \}}^{\{i-1,i\}} -  \tilde{\Phi}_{\{i-1 \}}^{\{i-1,i \}}\|_{\diamond} = O(e^{-\text{poly}(n)})$, which subsequently implies that
  \begin{equation}
    \|\rho_{\text{MPDO}} - \tilde{\rho}_{\text{MPDO}}\|_1 = O(e^{-\text{poly}(n)}).
  \end{equation}
  Therefore,
  \begin{equation}
    \begin{aligned}
      \|e^{-\beta H}/ \Tr(e^{-\beta H}) - \tilde{\rho}_{\text{MPDO}} \|_1 & \leq \|e^{-\beta H}/\Tr(e^{\beta H}) -\rho_{\text{MPDO}}\|_1 + \|\rho_{\text{MPDO}} - \tilde{\rho}_{\text{MPDO}} \|_1 \\
      & = O(\frac{1}{n^{\eta-1}}) + O(e^{-\text{poly(n)}}).
    \end{aligned}
  \end{equation}
\end{proof}

\subsection{Efficiently computable upper bound on free energy}
We show that there exists a Markivan MPDO with a polynomially large bond dimension that approximates the free energy of the Gibbs state up to an inverse polynomial error. This implies that for inverse polynomially small promise gap, there exists a polynomial-size classical proof of the ``YES'' instance that can be verified(in polynomial-time).
\begin{lem}
 Under Conjecture \ref{conjecture:Markov}, for any 1D local Hamiltonian $H$ and temperature $T$, for any constant $\alpha >0$, there exists a $O(\frac{1}{n^{\alpha+2}})$-approximate Markovian MPDO $\rho_{\text{MPDO}}$ with $\text{poly}(n)$ bond dimension specified with polynomial precision such that
  \begin{equation}
    \Tr(\rho H) - TS(\rho) \leq \Tr(\rho_{\text{MPDO}} H) -TS(\rho_{\text{MPDO}})+  \frac{CT}{n^{\alpha}},
  \end{equation}
  for some constant $C$,  where $\rho = e^{-\beta H} / \Tr(e^{-\beta H})$.
  \label{lemma:free_energy_bound}
\end{lem}
\begin{proof}
  Because of Theorem \ref{thm:finite_precision}, there exists a $O(\frac{1}{n^{\eta}})$-approximate Markovian MPDO $\rho_{\text{MPDO}}$ with $O(n^{2\eta e^{c\beta}/a})$-bond dimension and qudit dimension with polynomial precision such that $\delta = \|\rho - \rho_{\text{MPDO}} \|_1  = O(\frac{1}{n^{\eta - 1}})$ for everey $\eta$. Let us choose $\eta = 2+\alpha$. Because $H$ is a local Hamiltonian,
  \begin{equation}
    \begin{aligned}
      \Tr(\rho H - \rho_{\text{MPDO}} H) &\leq \|\rho - \rho_{\text{MPDO}} \|_1 \\
      &\leq n^{-(1+\alpha)}.
    \end{aligned}
  \end{equation}
  Due to the continuity of entropy\cite{Fannes1973},
  \begin{equation}
    \begin{aligned}
      |S(\rho) - S(\rho_{\text{MPDO}})| &\leq 2\delta n \log d - 2\delta \log (2\delta) \\
      &\leq 2n^{-\alpha}\log d  + O(n^{-1-\alpha} \log n).
    \end{aligned}
  \end{equation}
  These bounds yield the claim.
\end{proof}

\begin{thm}
  Under Conjecture \ref{conjecture:Markov}, the free energy problem for 1D local Hamiltonian at temperature $T= O(1)$ is in NP.
  \label{thm:NP}
\end{thm}
\begin{proof}
  Recall that the minimum of $\Tr(\rho H) - TS(\rho)$ is achieved by the Gibbs state, $\rho = e^{-\beta H} / \Tr(e^{-\beta H})$. By Lemma \ref{lemma:free_energy_bound}, for any constant $\alpha>0$, there exists a $\frac{1}{\text{poly}(n)}$-approximate Markovian MPDO, denoted as $\rho_{\text{MPDO}}$, such that its free energy differs from the minimum free energy by at most $CTn^{-\alpha}$.

  By Theorem \ref{thm:entropy_decomposition}, the entropy of $\rho_{\text{MPDO}}$ can be approximated by a linear combination of entropies over $1-$ and $2-$body density matrices up to $O(\frac{1}{n^{1+\alpha}})$ error. Because $\rho_{\text{MPDO}}$ is a MPDO with polynomially large bond dimension and qudit dimension, and because each of the entries are specified with polynomial precision, these density matrices can be computed in polynomial time. Thus, the following upper bound can be computed in polynomial time:
  \begin{equation}
    \Tr(\rho H) - TS(\rho) \leq \Tr(\rho_{MPDO} H)- T((\sum_{i=2}^{m} S(\rho_{\text{MPDO}}^{\{i-1,i \}})) - (\sum_{i=2}^{m-1}S(\rho_{\text{MPDO}}^{\{ i\}}))) + (m-2) (4\epsilon L_0 \log d + 2H_b(2\epsilon) ),
  \end{equation}
  where $m$ and $L_0$ are the ones that appear in the proof of Theorem \ref{thm:faithfulness} and
  \begin{equation}
    \epsilon = \max_i 
    \|\Phi_{\{i-1\}}^{\{i-1,i\}}(\tilde{\rho}_{\text{MPDO}}^{\{i-1\}}) - \Tr_{i+1}[(\mathcal{I}_{\{ i-1\}} \otimes \Phi_{\{i\}}^{\{i,i+1\}})\circ \Phi_{\{i-1\}}^{\{i-1,i\}}(\tilde{\rho}_{\text{MPDO}}^{\{i-1\}}) ] \|_1.
  \end{equation}
  These numbers, $\epsilon, m,$ and $L_0$ can be all computed in polynomial time and the resulting bound yields a $\frac{1}{\text{poly}(n)}$ error.

  Thus, there exists a Markovian MPDO which approximates the minimum free energy up to $\frac{1}{\text{poly}(n)}$ error. This upper bound can be certified efficiently because the fact that $\rho_{\text{MPDO}}$ is a $\frac{1}{\text{poly}(n)}$-approximate Markovian MPDO can be verified efficiently; see Proposition \ref{proposition:MMPDO_verification}.
\end{proof}

\section{Conclusion}
We found that the so called Markovian MPDO has a very special property: that its entropy can be efficiently computed. An important insight was that one can certify the fact that the global state forms a Markov chain from a set of local data. This is why we could show that the global entropy has a local decomposition and certify this fact efficiently.

We introduced the free energy problem, which is a natural generalization of the ground state energy problem to finite temperature. We showed that, modulo a certain reasonable conjecture, the 1D free energy problem at finite temperature is in NP. Unlike the ground state energy problem, which is QMA-complete, it is interesting to note that the same problem may become easy provided that the conjecture is true. It is important to note that no extra structure, e.g., the existence of spectral gap, is required to make this conclusion.

Our result supports a physical intuition that quantum coherence is lost at finite temperature, and thus makes the problem easier than the ground state problem. Of course, it should be noted that this conclusion is based on Conjecture \ref{conjecture:Markov}, which remains unproven. Proving this conjecture, as well as studying the complexity under such assumptions in higher dimensions would be the natural research direction to pursue. Another interesting question is whether the free energy problem for 1D local Hamiltonian is in P at finite temperature.

\bibliographystyle{ieeetr}
\bibliography{bib}

\end{document}